\documentclass[11pt,reqno]{amsart}

\usepackage{amsmath, amsthm, amstext, amssymb, amsfonts, graphicx, bm, subfigure}
\renewcommand{\L}{\mathcal{L}}

\newtheorem{theorem}{Theorem}[section]

\newtheorem{definition}[theorem]{Definition}
\newtheorem{lemma}[theorem]{Lemma}
\newtheorem{proposition}[theorem]{Proposition}
\newtheorem*{remark}{Remark}

\begin{document}

\title{A cluster expansion approach to renormalization group transformations}

\author{Mei Yin}
\address{Department of Mathematics, University of Texas, Austin, TX
78712, USA} \email{myin@math.utexas.edu}

\dedicatory{\rm May 19, 2011}

\begin{abstract}
The renormalization group (RG) approach is largely responsible for
the considerable success which has been achieved in developing a
quantitative theory of phase transitions. This work treats the
rigorous definition of the RG map for classical Ising-type lattice
systems in the infinite volume limit at high temperature. A
cluster expansion is used to justify the existence of the partial
derivatives of the renormalized interaction with respect to the
original interaction. This expansion is derived from the formal
expressions, but it is itself well-defined and convergent. Suppose
in addition that the original interaction is finite-range and
translation-invariant. We will show that the matrix of partial
derivatives in this case displays an approximate band property.
This in turn gives an upper bound for the RG linearization.
\end{abstract}

\maketitle

\section{Introduction}
\label{intro} We consider renormalization group (RG)
transformations for Ising-type lattice spin systems on
$\mathbb{Z}^d$. The spins in the original lattice $\L$ are denoted
by $\sigma$, whereas the block spins in the image lattice $\L'$
are denoted by $\sigma'$, and assumed to be of Ising-type also.
This assumption allows for treatment of many important RG
transformations such as decimation and majority rule, but is not
applicable for more general types such as block-average
transformations. $\L'$ indexes a partition of $\L$ into blocks,
all with the same cardinality $s$. Thus for each site $y$ in
$\L'$, there is a corresponding block $y^o$ that is a subset of
$\L$. Also, $\L$ is endowed with a metric $d$, and this naturally
induces a metric $d'$ on $\L'$.

Formally, the RG maps a Hamiltonian $H(\sigma)=-\sum_X
J(X)\sigma_X$ into a renormalized Hamiltonian $H'(\sigma')=-\sum_Y
J'(Y)\sigma'_Y$:
\begin{equation}
e^{-H'(\sigma')}=\sum_{\sigma}\prod_{y\in
\L'}T_y(\sigma,\sigma'_y)e^{-H(\sigma)},
\end{equation}
where $\sum_{\sigma}$ and $\sum_{\sigma'}$ (normalized sums)
denote the product probability measures on $\{+1,-1\}^{\L}$ and
$\{+1,-1\}^{\L'}$, respectively, and $T_y(\sigma, \sigma'_y)$
denotes a specific RG probability kernel, which depends only on
$\sigma$ through the block corresponding to $y$, and satisfies
both a symmetry condition,
\begin{equation}
\label{sym} T_y(\sigma, \sigma'_y)=T_y(-\sigma, -\sigma'_y),
\end{equation}
and a normalization condition,
\begin{equation}
\label{norm} \sum_{\sigma'}T_y(\sigma, \sigma'_y)=1
\end{equation}
for every $\sigma$ and every $y$. Notice that because of
(\ref{sym}) and (\ref{norm}),
\begin{equation}
\label{begin} \sum_{\sigma}T_y(\sigma,
+1)=\sum_{\sigma}T_y(\sigma, -1)=1.
\end{equation}
Equation (\ref{begin}) will be of fundamental importance for the
purpose of this paper (cf. proof of Proposition \ref{par}),
whereas assumptions (\ref{sym}) and (\ref{norm}), which lead to
(\ref{begin}), are not essential. This, however, would rule out RG
transformations where different block spins have unequal
occurrence probabilities. (Some references for this literature may
be found, for instance, in a series of lectures by Brydges
\cite{Brydges}.)

Our basic assumption is that the original interaction $J$ lies in
a Banach space $\mathcal{B}_r$, with norm
\begin{equation}
\label{J} ||J||_r=\sup_{x \in \L}\sum_{X: x \in X}|J(X)|e^{r|X|},
\end{equation}
where the constant $r>0$ and $|X|$ denotes the cardinality of the
set $X$. The properties of the RG transformation have been studied
extensively by mathematical physicists over a period of many
years. The first existence results of renormalized interactions in
trivial unique-phase regimes were obtained by Griffiths and Pearce
\cite{Griffiths} in the low density (or high magnetic field)
regime by cluster expansion methods. They also presented plausible
arguments showing that the RG transformation exhibits a rather
peculiar behavior at low temperatures. Robert Israel \cite{Israel}
justified the existence of the renormalized interaction $J'$ using
beautiful and ingenious techniques involving Banach algebras and
conditional expectations. Kashapov \cite{Kashapov} worked with
cumulants (semi-invariants), with estimates that relied on
combinatorial methods of Malyshev \cite{Malyshev}, and showed that
the RG map can be formalized rigorously in terms of the
Hamiltonian of a Gibbs field. Martinelli and Olivieri
\cite{Mar1}$^{,}$\cite{Mar2} investigated the stability and
instability of pathologies of RG transformations under decimation.
For block-average RG transformations, Cammarota \cite{Cammarota}
proved that the block spin interaction tends in norm to a one-body
quadratic potential in the infinite volume limit at high
temperature, whereas van Enter \cite{vanEnter1} constructed a
counterexample showing that the renormalized measure may not be
Gibbsian even when the temperature is above the critical
temperature.

These results are extended by the recent work \cite{Yin1} which
analyzes the spectrum of the RG maps corresponding to decimation
and majority rule at infinite temperature and discovers that it is
of an unusual kind: dense point spectrum for which the adjoint
operators have no point spectrum at all, but only residual
spectrum. The present investigation is a follow-up to my previous
work and explores various existence properties of the RG at high
temperature. It employs a reasonably straightforward application
of the cluster expansion machinery and justifies the existence of
the partial derivatives of the renormalized interaction $J'$ with
respect to the original interaction $J$ (Theorem \ref{expar}).
Under the additional assumption that the original interaction $J$
is finite-range and translation-invariant, it will be shown that
the matrix of partial derivatives displays an approximate band
property (Theorem \ref{band}). This in turn gives an upper bound
for the RG linearization (Theorem \ref{lin}).

The real interest of the RG is to define the transformation at
intermediate temperature, in particular, the critical temperature.
This is a considerably more difficult enterprise: one could worry
about the many issues raised by van Enter, Fern\'{a}ndez, and
Sokal \cite{vanEnter2}. Fortunately, there is some hope for
progress in this area due to the fact that the correlation length
of the constrained system relevant to the definition of the RG
transformation may well be finite, and may even sometimes be used
as a small parameter. Pioneering efforts were made by Olivieri and
his various collaborators \cite{Oliv}$^{,}$\cite{Olivieri}.
Another approach in a similar spirit was developed in the
important work of Kennedy \cite{Kennedy}. Additional references
may be found in
\cite{Aizenman}$^{,}$\cite{Faris}$^{,}$\cite{Fernandez}$^{,}$\cite{Lorinczi}.
Using similar techniques as in the present work \cite{Yin2}, it
may be shown that parallel results for RG linearization hold under
the condition proposed by Haller and Kennedy \cite{Haller}.

Another possible generalization of this work is in the context of
Potts models. Working with transmissivities rather than coupling
constants \cite{Tsallis}, results on spectral properties of the RG
map are expected in the high temperature regime through a more
involved application of the cluster expansion machinery.

\section{Cluster expansion}
\label{renormalized} This section gives cluster expansion
expressions which are valid for finite lattices.

\begin{proposition}
\label{expand} For every subset $W$ of the original lattice and
every subset $Z$ of the image lattice, the partial derivative
$\frac{\partial J'(Z)}{\partial J(W)}$ of the RG transformation is
given by the expression
\begin{equation}
\label{part} \frac{\partial J'(Z)}{\partial
J(W)}=\sum_{\sigma'}\sigma'_Z \frac{\sum_{\sigma}\prod_{y \in
\L'}T_{y}(\sigma,\sigma'_{y})e^{\sum_X
J(X)\sigma_X}\sigma_W}{\sum_{\sigma}\prod_{y \in
\L'}T_{y}(\sigma,\sigma'_{y})e^{\sum_X J(X)\sigma_X}}.
\end{equation}
\end{proposition}

\begin{proof}
By the use of the Fourier series on the group $\{+1, -1\}^{\L'}$,
we see that the renormalized coupling constants $J'$ are given by
\begin{equation}
\label{JZ} J'(Z)=\sum_{\sigma'}\sigma'_Z
\log\left(\sum_{\sigma}\prod_{y \in
\L'}T_{y}(\sigma,\sigma'_{y})e^{\sum_X J(X)\sigma_X}\right).
\end{equation}
We then take the derivative of both sides of (\ref{JZ}) with
respect to $J(W)$.
\end{proof}

To understand the following Proposition II.2, we need to introduce
some combinatorial concepts. A hypergraph is a set of sites
together with a collection $\Gamma$ of nonempty subsets. Such a
nonempty set is referred to as a hyper-edge or link. Two links are
block-connected if they both intersect some block. The support of
a hypergraph is the set $\cup \Gamma$ of sites that belong to some
set in $\Gamma$. A hypergraph $\Gamma$ is block-connected if the
support of $\Gamma$ is nonempty and cannot be partitioned into
nonempty sets with no block-connected links. In our current
setting, a subset $X$ of $\mathcal{L}$ defines a subset $X'$ of
$\mathcal{L'}$, corresponding to the set of blocks that have
non-empty intersection with $X$. Thus a hypergraph $\Gamma$ on
$\mathcal{L}$ defines a hypergraph $\Gamma'$ on $\mathcal{L'}$. We
use $\Gamma_c$ to indicate block connectivity of the hypergraph
$\Gamma_c$, and write $\Gamma_c^*=\cup \Gamma'_c$ for the support
of $\Gamma'_c$ in the image lattice.

\begin{proposition}
\label{par} Let $W(\sigma')$ be the frozen-block-spin partition
function
\begin{equation}
W(\sigma')=\sum_{\sigma}\prod_{y \in
\L'}T_{y}(\sigma,\sigma'_{y})e^{\sum_X J(X)\sigma_X}.
\end{equation}
For fixed values of the renormalized spins $\sigma'$, $W(\sigma')$
has the cluster representation
\begin{equation}
\label{tt} W(\sigma')=\sum_{\Delta}\prod_{N\in \Delta}w_N,
\end{equation}
where $\Delta$ is a set of disjoint subsets $N$'s of $\L'$, and
\begin{equation}
\label{w} w_N=\sum_{\Gamma_c^*=N}\alpha(N, \Gamma_c,
\{\sigma'\}_N),
\end{equation}
and the sum here is over block-connected hypergraphs $\Gamma_c$ on
$\mathcal{L}$ whose images in $\mathcal{L'}$ have support $N$. The
contribution of each block-connected hypergraph is given by
\begin{equation}
\label{alpha} \alpha(N, \Gamma_c,
\{\sigma'\}_N)=\sum_{\sigma}\prod_{y\in
N}T_y(\sigma,\sigma'_y)\prod_{X\in \Gamma_c}
\left(e^{J(X)\sigma_{X}}-1\right).
\end{equation}
\end{proposition}

\begin{proof}
When the original Hamiltonian $H$ is at high temperature
($||J||_r$ small), we can rewrite $e^{\sum_X J(X)\sigma_X}$ as a
perturbation around zero interaction (infinite temperature),
\begin{eqnarray}
\label{t} W(\sigma')&=&\sum_{\sigma}\prod_{y\in
\mathcal{L'}}T_y(\sigma,\sigma'_y)\prod_X
\left(1+e^{J(X)\sigma_X}-1\right)\notag\\&=&\sum_{\sigma}\prod_{y\in
\mathcal{L'}} T_y(\sigma,\sigma'_y)\sum_{\Gamma}\prod_{X\in
\Gamma} \left(e^{J(X)\sigma_{X}}-1\right),
\end{eqnarray}
where $\Gamma$ is a set of subsets $X$'s of $\mathcal{L}$.

We are going to organize the sum over hypergraphs in (\ref{t}) in
the following way. Each hypergraph $\Gamma$ on $\mathcal{L}$ has a
support in $\mathcal{L'}$, which breaks up into block-connected
parts. Let $\Delta$ be the parts, and for $N\in \Delta$, let
$S(N)$ be the corresponding block-connected hypergraph on this
part, i.e., $S(N)^*=N$. Then summing over hypergraphs $\Gamma$ is
equivalent to summing over $\Delta$ and functions $S$ with the
appropriate property. Furthermore, the product over $N$ in
$\Delta$ and the links in $S(N)$ is equivalent to the product over
the corresponding $\Gamma$. We have
\begin{equation}
W(\sigma')=\sum_{\sigma}\prod_{y\in
\mathcal{L'}}T_y(\sigma,\sigma'_y)\sum_{\Delta}\sum_S\prod_{N \in
\Delta}\prod_{X\in S(N)}\left(e^{J(X)\sigma_{X}}-1\right).
\end{equation}
By independence, the sum over $\sigma$ can be factored over
$\Delta$, and this gives
\begin{equation}
\label{comp} W(\sigma')=\sum_{\Delta}\prod_{y \notin \cup
\Delta}\sum_{\sigma}T_y(\sigma,\sigma'_y)\sum_S \prod_{N \in
\Delta}\sum_{\sigma}\prod_{y\in N}T_y(\sigma,\sigma'_y)\prod_{X\in
S(N)} \left(e^{J(X)\sigma_{X}}-1\right).
\end{equation}
Notice that because of (\ref{begin}), many of the $T_y$ factors
sum to $1$, (\ref{comp}) can be simplified,
\begin{equation}
W(\sigma')=\sum_{\Delta}\sum_S \prod_{N \in \Delta}\alpha(N, S(N),
\{\sigma'\}_N).
\end{equation}
And by the distributive law,
\begin{equation}
\sum_S \prod_{N \in \Delta}\alpha(N, S(N),
\{\sigma'\}_N)=\prod_{N\in \Delta}\sum_{\Gamma_c^*=N}\alpha(N,
\Gamma_c, \{\sigma'\}_N).
\end{equation}
Therefore
\begin{eqnarray}
\label{eq} W(\sigma')&=&\sum_{\Delta}\prod_{N\in
\Delta}\sum_{\Gamma_c^*=N}\alpha(N, \Gamma_c, \{\sigma'\}_N).
\end{eqnarray}
Our claim thus follows.
\end{proof}

We rewrite (\ref{tt}) in the following way to apply standard
results on cluster expansion,
\begin{eqnarray}
\sum_{\Delta}\prod_{N\in \Delta} w_N\notag
&=&\sum_{p=0}^{\infty}\frac{1}{p!}\sum_{N_1,...,N_p}\prod_{\{i,j\}}\left(1-c(N_i,N_j)\right)w_{N_1}\cdots w_{N_p}\notag\\
&=&\sum_{p=0}^{\infty}\frac{1}{p!}\sum_{N_1,...,N_p}\sum_{G}\prod_{\{i,j\}\in
G}\left(-c(N_i,N_j)\right)w_{N_1}\cdots w_{N_p},
\end{eqnarray}
where $G$ is a graph with vertex set $\{1,...,p\}$ and
\begin{eqnarray}
\label{c} c(N_i,N_j)=\left\{\begin{array}{ll}
1 & \mbox{if $N_i$ and $N_j$ overlap};\\
0 & \mbox{otherwise}.\end{array} \right.
\end{eqnarray}

\begin{proposition}
The frozen-block-spin free energy $\log(W(\sigma'))$ is given by
the cluster expansion
\begin{equation}
\log(W(\sigma'))=\sum_{p=1}^{\infty}\frac{1}{p!}\sum_{N_1,...,N_p}C\left(N_1,...,N_p\right)w_{N_1}\cdots
w_{N_p},
\end{equation}
where
\begin{equation}
\label{C} C\left(N_1,...,N_p\right)=\sum_{G_c}\prod_{\{i,j\}\in
G_c}\left(-c(N_i,N_j)\right),
\end{equation}
and $G_c$ is a connected graph with vertex set $\{1,...,p\}$.
\end{proposition}

\begin{proof}
The effect of taking the logarithm is that the sum over graphs is
replaced by the sum over connected graphs:
\begin{equation}
\log(W(\sigma'))=
\sum_{p=1}^{\infty}\frac{1}{p!}\sum_{N_1,...,N_p}\sum_{G_c}\prod_{\{i,j\}\in
G_c}\left(-c(N_i,N_j)\right)w_{N_1}\cdots w_{N_p}.
\end{equation}
Our claim thus follows.
\end{proof}

\section{Existence of the partial derivatives}
From now on, we work in the infinite volume limit of the lattice
system. Following standard interpretation of statistical
mechanics, our results hold when the original interaction $J$ is
at high temperature.

\label{derivative}
\begin{theorem}[Koteck\'{y}-Preiss]
\label{KP} Recall (\ref{c}) and (\ref{C}). Take $1<M<e^r$. Suppose
that
\begin{equation}
\label{conv} \sum_{N'}c(N, N')|w_{N'}|M^{|N'|}\leq |N|\log(M).
\end{equation}
Then the avoidance probability for every $Y\subset \L'$ has a
convergent power series expansion,
\begin{equation*}
\left|\sum_{\Delta'}\prod_{N\in\Delta'}w_N/\sum_{\Delta}\prod_{N\in
\Delta} w_N\right|
\end{equation*}
\begin{eqnarray*}
=\left|\exp\left(-\sum_{p=1}^{\infty}\frac{1}{p!}\sum_{N_1,...,N_p}C\left(N_1,...,N_p\right)c(Y,
\cup_1^p N_i)w_{N_1}\cdots w_{N_p}\right)\right|
\end{eqnarray*}
\begin{eqnarray}
\label{avoid} \leq
\exp\left(\sum_{p=1}^{\infty}\frac{1}{p!}\sum_{N_1,...,N_p}|C\left(N_1,...,N_p\right)|c(Y,
\cup_1^p N_i)|w_{N_1}|\cdots |w_{N_p}|\right)\leq M^{|Y|},
\end{eqnarray}
where $\Delta'$ is a set of disjoint subsets of $\L'
\texttt{\char92} Y$, and $\Delta$ is a set of disjoint subsets of
$\L'$. Notice that here we are only counting contributions of
block-connected $N_i$'s that are also block-connected to $Y$.
\end{theorem}

\begin{proposition}
\label{basic} Take $1<M<e^r$. Consider the original coupling
constants $J$ with the Banach space norm $||J||_r$. Suppose $J$ is
at high temperature ($||J||_r$ small),
\begin{equation}
\label{Jr}||J||_r \leq
\frac{\log(M)c^2}{2s\left(c+\log(M)\right)},
\end{equation}
where $c=\frac{1}{\sqrt{\epsilon}}-1$ and $\epsilon=Me^{-r}$. Then
for each block spin site $y$, we have
\begin{equation}
\label{ineq} \sum_{y\in N}|w_N| M^{|N|}\leq \log(M).
\end{equation}
\end{proposition}

\begin{remark}
\textnormal{The inequality (\ref{ineq}) is a standard sufficient
condition for (\ref{conv}). It will be applied in the following
Theorem \ref{expar}.}
\end{remark}

\begin{proof}
We notice that when $||J||_r$ is small (say $||J||_r\leq
\frac{1}{2}$), $e^{|J(X)|}-1\leq 2|J(X)|$ by the mean value
theorem. Also, it easily follows from (\ref{J}) that for all $X$
with cardinality $m$ and containing a fixed $x$, $\sum|J(X)|\leq
||J||_re^{-rm}$. More importantly, for $\Gamma_c^*=N$, $|N|\leq
\sum|X|$ with $X$ in $\Gamma_c$. We have
\begin{eqnarray}
\sum_{y\in N}|w_N| M^{|N|}&\leq&\sum_{y\in N}\sum_{\Gamma_c^*=N}M^{|N|}\prod_{X\in \Gamma_c}2|J(X)|\notag\\
&\leq&\sum_{y\in N}\sum_{\Gamma_c^*=N}\prod_{X\in
\Gamma_c}2|J(X)|M^{|X|}\notag\\&=&\sum_{y\in
\Gamma_c^*}\prod_{X\in \Gamma_c}2|J(X)|M^{|X|}.
\end{eqnarray}
We say that a hypergraph $\Gamma_c$ is block-rooted at $y$ if its
support intersects a fixed block $y^o$. Let $a_n(y)$ be the
contribution of all block-connected hypergraphs with $n$ links
that are block-rooted at $y$,
\begin{equation}
\label{ann} a_n(y)=\sum_{y\in \Gamma_c^*: |\Gamma_c|=n}\prod_{X\in
\Gamma_c}2|J(X)|M^{|X|}.
\end{equation}
Then
\begin{equation}
\sum_{y\in N}|w_N| M^{|N|}\leq \sum_{n=1}^{\infty}\sup_{y\in
\mathcal{L'}}a_n(y).
\end{equation}

Let $a_n$ be the supremum over $y$ of the contribution of
block-connected hypergraphs with $n$ links that are block-rooted
at $y$, i.e., $a_n=\sup_{y\in \mathcal{L'}}a_n(y)$. It seems that
once we show that $a_n$ is exponentially small, the geometric
series above will converge, and our claim might follow. To
estimate $a_n$, we relate to some standard combinatorial facts
\cite{Minlos}. The rest of the proof follows from a series of
lemmas.
\end{proof}

\begin{lemma}
Let $a_n$ be the supremum over $y$ of the contribution of
block-connected hypergraphs with $n$ links that are block-rooted
at $y$. Then $a_n$ satisfies the recursive bound
\begin{equation}
a_n\leq
2s||J||_r\sum_{m=1}^{\infty}\epsilon^m\sum_{k=0}^{m}\tbinom
mk\sum_{a_{n_1},...,a_{n_k}: n_1+\cdots+n_k+1=n}a_{n_1}\cdots
a_{n_k}
\end{equation}
for $n\geq 1$, where $\tbinom mk$ is the binomial coefficient.
\end{lemma}

\begin{proof}
We first linearly order the points $x$ in $\mathcal{L}$ and also
linearly order the subsets $X$ of $\mathcal{L}$. This naturally
induces a linear ordering of the points $y$ in $\mathcal{L'}$. For
a fixed but arbitrarily chosen $y$ in $\mathcal{L'}$, we examine
(\ref{ann}). Write $\Gamma_c=\{X_1\}\cup \Gamma^1_c$, where $X_1$
is the least $X$ in $\Gamma_c$ with $y^o \cap X_1 \neq \emptyset$.
There must be such a set, since $y\in \Gamma_c^*$. Moreover, there
must be some $x\in y^o$ such that $x\in X_1$, of which there are
$s$ possibilities, as the block cardinality is $s$. Then
\begin{equation}
a_n(y)\leq s\sum_{m=1}^{\infty}\sup_{x\in \mathcal{L}}\sum_{X_1:
x\in X_1, |X_1|=m}2|J(X_1)|M^{m}\sum_{\Gamma^1_c}\prod_{X\in
\Gamma^1_c}2|J(X)|M^{|X|}.
\end{equation}
As a consequence,
\begin{equation}
a_n(y)\leq \sum_{m=1}^{\infty}2s||J||_r
\epsilon^{m}\sum_{\Gamma^1_c}\prod_{X\in
\Gamma^1_c}2|J(X)|M^{|X|}.
\end{equation}
The remaining hypergraph $\Gamma^1_c$ has $n-1$ subsets and breaks
into $k: k\leq m$ block-connected components
$\Gamma_1,...,\Gamma_k$ of sizes $n_1,...,n_k$, with
$n_1+\cdots+n_k=n-1$. The set of such components may be empty, or
it could just be the original block-connected set. For each
component $\Gamma_i$, there is a least block $y_i^o$ through which
it is block-connected to $X_1$. The image $\{y_i\}$ of these
blocks is a subset of $X'_1$ in $\mathcal{L'}$, thus has no more
than $k: k\leq m$ points, and the components are block-rooted at
these image sites. Furthermore, different $\Gamma_i$'s correspond
to disjoint $\Gamma'_i$'s, as $y_i\in \Gamma'_i$, the map from the
components to this image is injective. So we have
\begin{equation}
a_n(y)\leq \sum_{m=1}^{\infty}2s||J||_r \epsilon^{m}\sum_{k=0}^{m}
\tbinom{m}{k}\sum_{a_{n_1},...,a_{n_k}:
n_1+\cdots+n_k+1=n}a_{n_1}\cdots a_{n_k}.
\end{equation}
Our inductive claim follows by taking the supremum over all $y$ in
$\mathcal{L'}$. Finally, we look at the base step: $n=1$. In this
simple case, as reasoned above, we have
\begin{eqnarray}
a_1&=&\sup_{y\in \mathcal{L'}}\sum_{y\in \Gamma_c^*:
|\Gamma_c|=1}\prod_{X\in
\Gamma_c}2|J(X)|M^{|X|}\notag\\&\leq&s\sum_{m=1}^{\infty}\sup_{x\in
\mathcal{L}}\sum_{X: x\in X,
|X|=m}2|J(X)|M^{m}\notag\\&=&\sum_{m=1}^{\infty}2s||J||_r\epsilon^{m},
\end{eqnarray}
and this verifies our claim.
\end{proof}

Clearly, $\sum_{y\in N}|w_N| M^{|N|}$ will be bounded above by
$\sum_{n=1}^{\infty}\bar{a}_n$, if
\begin{equation}
\label{a} \bar{a}_n=
2s||J||_r\sum_{m=1}^{\infty}\epsilon^m\sum_{k=0}^m\tbinom
mk\sum_{\bar{a}_{n_1},...,\bar{a}_{n_k}:
n_1+\cdots+n_k+1=n}\bar{a}_{n_1}\cdots \bar{a}_{n_k}
\end{equation}
for $n\geq 1$, i.e., equality is obtained in the above lemma.

\begin{lemma}
Consider the coefficients $\bar{a}_n$ that bound the contributions
of block-connected and block-rooted hypergraphs with $n$ links.
Let $w=\sum_{n=1}^{\infty}\bar{a}_n z^n$ be the generating
function of these coefficients. Then the recursion relation
(\ref{a}) for the coefficients is equivalent to the formal power
series generating function identity
\begin{equation}
\label{id}
w=2s||J||_rz\sum_{m=1}^{\infty}\epsilon^m(1+w)^m=2s||J||_rz\frac{\epsilon(1+w)}{1-\epsilon(1+w)}.
\end{equation}
\end{lemma}

\begin{proof}
Notice that $(1+w)^m=\sum_{k=0}^m \tbinom mk w^k$, thus
\begin{eqnarray}
w=2s||J||_rz\sum_{m=1}^{\infty}\epsilon^m\sum_{k=0}^m \tbinom mk
w^k.
\end{eqnarray}
Writing completely in terms of $z$, we have
\begin{equation}
\sum_{n=1}^{\infty}\bar{a}_n
z^n=2s||J||_r\sum_{m=1}^{\infty}\epsilon^m\sum_{k=0}^m \tbinom mk
\sum_{\bar{a}_{n_1},...,\bar{a}_{n_k}:
n_1+\cdots+n_k+1=n}\bar{a}_{n_1}\cdots \bar{a}_{n_k}z^n.
\end{equation}
Our claim follows from term-by-term comparison.
\end{proof}

\begin{lemma}
If $w$ is given as a function of $z$ as a formal power series by
the generating function identity (\ref{id}), then this power
series has a nonzero radius of convergence $|z|\leq
\frac{1}{2s||J||_r}c^2$. For big enough $c$, this radius of
convergence is arbitrarily large, and in particular, the series
will converge for $z=1$, i.e., the sum of the bounds on the
contributions of block-connected and block-rooted hypergraphs
converges.
\end{lemma}

\begin{proof}
Without loss of generality, assume $z\geq 0$. Set $z_1=2s||J||_r
z$. Solving (\ref{id}) for $z_1$ gives
\begin{equation}
z_1=\frac{w\left(1-\epsilon(1+w)\right)}{\epsilon(1+w)}.
\end{equation}
By elementary calculus, this increases as $w$ goes from $0$ to $c$
to have values $z_1$ from $0$ to $c^2$. It follows that as $z_1$
goes from $0$ to $c^2$, the $w$ values range from $0$ to $c$.
\end{proof}

\noindent \textbf{Proof of Proposition \ref{basic} continued.} We
notice that in the above lemma,
$w=\sum_{n=1}^{\infty}\bar{a}_nz^n=c$ corresponds to
$z_1=2s||J||_r z=c^2$, which implies that for each $n$,
\begin{equation}
\label{imply} \bar{a}_n \leq c(2s||J||_r)^nc^{-2n}.
\end{equation}
Gathering all the information we have obtained so far,
\begin{eqnarray}
\sum_{y\in N}|w_N| M^{|N|}&\leq&\sum_{n=1}^{\infty}c(2s||J||_r)^nc^{-2n}\notag\\
&=&c\frac{\frac{2s||J||_r}{c^2}}{1-\frac{2s||J||_r}{c^2}}\leq
\log(M)
\end{eqnarray}
by (\ref{Jr}). \qed

We have shown in (\ref{tt}) that the denominator of (\ref{part})
has a cluster representation. We now examine the effect of
multiplying $\sigma_W$ to this cluster representation as in the
numerator of (\ref{part}). There will be two kinds of terms. In
some of these, none of the block-connected components intersect
$W$, so for these terms one gets a product of $\sigma_W$ with a
product of independent $w_N$'s. For the other terms one decomposes
$\Delta$ into one block-connected component that is connected to
$W$ and remaining block-connected components that are not. The
result is the representation $\sum_{R,
\Delta'}\tilde{w}_R\prod_{N\in\Delta'}w_N$, where $R=\emptyset$ or
$R\cap W'\neq \emptyset$, and $\tilde{w}_R$ is a sum over
hypergraphs $\Delta_R$ with $\cup \Delta_R=R$ such that $W$,
$\Delta_R$ is block-connected. Therefore
\begin{equation}
\label{frac} \frac{\partial J'(Z)}{\partial
J(W)}=\sum_{\sigma'}\sigma'_Z\frac{\sum_{R,
\Delta'}\tilde{w}_R\prod_{N\in\Delta'}w_N}{\sum_{\Delta}\prod_{N\in
\Delta} w_N}.
\end{equation}

\begin{theorem}
\label{expar} Suppose the original interaction $J$ is at high
temperature (cf. (\ref{Jr})). Then for every subset $W$ of the
original lattice and every subset $Z$ of the image lattice, the
partial derivative $\frac{\partial J'(Z)}{\partial J(W)}$ of the
RG transformation (\ref{frac}) is well-defined.
\end{theorem}

\begin{proof}
The proof of this theorem is an application of the
Koteck\'{y}-Preiss result \cite{Kotecky}. Recall that $N\in
\Delta'$ implies $N\cap(R\cup W')=\emptyset$. By (\ref{avoid}),
\begin{equation}
\label{MWR}
\left|\sum_{\Delta'}\prod_{N\in\Delta'}w_N/\sum_{\Delta}\prod_{N\in
\Delta} w_N\right|\leq M^{|R\cup W'|}.
\end{equation}
To verify our claim, we need to estimate
\begin{equation}
\label{estimate} |\frac{\partial J'(Z)}{\partial J(W)}| \leq
\sum_{R}|\tilde{w}_R|M^{|R\cup W'|} \leq
\sum_{\Delta_R}M^{|W|}\prod_{Y\in \Delta_R}|w_Y| M^{|Y|}.
\end{equation}
But this is easy, remove $W$, the remaining hypergraph breaks up
into $k: 0\leq k\leq |W|$ block-connected components. So this last
quantity is bounded by
\begin{equation}
\label{suc} M^{|W|}\sum_{k=0}^{|W|}\tbinom {|W|}{k} \left(\log
(M)\right)^k=M^{|W|} \left(1+\log(M)\right)^{|W|}.
\end{equation}
\end{proof}

\section{Band structure}
\label{band} In this section, we concentrate our attention on
finite-range and translation-invariant Hamiltonians. We will show
that the matrix of partial derivatives in this case displays an
approximate band property.

\begin{remark}
\textnormal{Let
\begin{equation}
\text{diam}(X)=\sup\{d(x,y): x\in X, y\in X\}
\end{equation}
be the volume of a subset $X$ of the original lattice. For later
purposes, we point out that the finite-range assumption on the
Hamiltonian implies a weaker assumption, finite-body, i.e., there
is a constant $S$ such that $J(X)=0$ for $\text{diam}(X)>S$
implies there is a constant $D$ such that $J(X)=0$ for $|X|>D$,
where $D$ only depends on the maximum possible range $S$ and the
number of dimensions $d$.}
\end{remark}

\begin{proposition}
\label{pin} Suppose the original interaction $J$ is at high
temperature (cf. (\ref{Jr})). Then for each block spin site $y$,
we have
\begin{equation}
\sum_{y\in N: |N|>P}|w_N| M^{|N|}\leq \epsilon(P),
\end{equation}
where \begin{equation} \label{epsilon}
\epsilon(P)=c\frac{\left(\frac{2s||J||_r}{c^2}\right)^{\frac{P}{D}}}{1-\frac{2s||J||_r}{c^2}},
\end{equation}
and for a fixed $||J||_r$ that satisfies (\ref{Jr}), $\epsilon(P)
\rightarrow 0$ as $P \rightarrow \infty$.
\end{proposition}

\begin{proof}
Due to the finite-body assumption on the Hamiltonian, any
block-connected hypergraph that is block-rooted at $y$ and with
cardinality greater than $P$ will have at least $P/D$ links. By
(\ref{imply}), this implies
\begin{eqnarray}
\sum_{y\in N: |N|>P}|w_N| M^{|N|}\leq
\sum_{n=P/D}^{\infty}c(2s||J||_r)^nc^{-2n}=\epsilon(P)
\end{eqnarray}
\end{proof}

\begin{proposition}
\label{pindown} Suppose the original interaction $J$ is at high
temperature (cf. (\ref{Jr})). Then for every $Y\subset \L'$, we
have
\begin{eqnarray}
\label{nfinal} \sum_{p=1}^{\infty}\frac{1}{p!}\sum_{N_1,...,N_p:
|\cup_1^p N_i|>P}|C\left(N_1,...,N_p\right)|c(Y, \cup_1^p
N_i)|w_{N_1}|\cdots |w_{N_p}| \leq |Y|\epsilon(P).
\end{eqnarray}
\end{proposition}

\begin{proof}
This follows from Proposition \ref{pin}. Remove $Y$, the remaining
hypergraph is still block-connected by (\ref{C}). Moreover, there
can be at most $|Y|$ choices for where it is pinned down.
\end{proof}

\begin{theorem}
\label{band} Suppose the original interaction $J$ is at high
temperature (cf. (\ref{Jr})). Suppose also that it is finite-range
and translation-invariant. Then there is an approximate band
property for the matrix of partial derivatives: For subset $W$ of
the original lattice and subset $Z$ of the image lattice that are
sufficiently far apart, the partial derivative $\frac{\partial
J'(Z)}{\partial J(W)}$ of the RG transformation (\ref{frac}) is
arbitrarily small.
\end{theorem}

\begin{remark}
\textnormal{Recall (\ref{epsilon}). Let
\begin{equation}
l(W, Z)=\inf\{d'(w, z): w\in W', z\in Z\}
\end{equation}
be the distance between $W$ and $Z$ measured in the image lattice.
If
\begin{equation} l(W, Z)>S(|W|P+QK),
\end{equation}
then
\begin{equation}
\label{enough} |\frac{\partial J'(Z)}{\partial J(W)}| \leq
M^{D}\left(1+\log(M)\right)^{D}\left(\frac{\epsilon(P)}{\log(M)}+\left(\epsilon(Q)+\epsilon(K)\right)(1+P)DM^{(1+P)D}\right).
\end{equation}}
\end{remark}

\begin{proof}
Fix a $P$ that is large enough. We rewrite (\ref{frac}) as
\begin{equation}
\label{small} \frac{\partial J'(Z)}{\partial
J(W)}=\sum_{\sigma'}\sigma'_Z\frac{\sum_{|R|>|W|P,
\Delta'}\tilde{w}_R\prod_{N\in\Delta'}w_N}{\sum_{\Delta}\prod_{N\in
\Delta} w_N}+\sum_{\sigma'}\sigma'_Z\frac{\sum_{|R|\leq |W|P,
\Delta'}\tilde{w}_R\prod_{N\in\Delta'}w_N}{\sum_{\Delta}\prod_{N\in
\Delta} w_N}.
\end{equation}
Following, we will verify the smallness of (\ref{small}) by
examining the two terms on the right-hand side separately.

\noindent Case 1: $|R|>|W|P$. Similarly as in the proof of Theorem
\ref{expar}, we estimate (\ref{estimate}). Remove $W$, the
remaining hypergraph (with cardinality greater than $|W|P$) breaks
up into $k: 0\leq k\leq |W|$ block-connected components, so at
least one of them has cardinality greater than $P$. By
(\ref{epsilon}), the contribution of this hypergraph is bounded by
\begin{equation}
M^{|W|}\epsilon(P)\sum_{k=0}^{|W|}\tbinom {|W|}{k} \left(\log
(M)\right)^{k-1}=\frac{\epsilon(P)}{\log(M)}M^{|W|}
\left(1+\log(M)\right)^{|W|}.
\end{equation}

\noindent Case 2: $|R|\leq |W|P$. We need to do a more careful
analysis for this case. By the Koteck\'{y}-Preiss theorem
\cite{Kotecky}, (\ref{ineq}) implies
\begin{multline}
\label{exp}
\sum_{\Delta'}\prod_{N\in\Delta'}w_N/\sum_{\Delta}\prod_{N\in
\Delta}
w_N\\=\exp\left(-\sum_{p=1}^{\infty}\frac{1}{p!}\sum_{N_1,...,N_p}C\left(N_1,...,N_p\right)c(R\cup
W', \cup_1^p N_i)w_{N_1}\cdots w_{N_p}\right).
\end{multline}
For notational convenience, we will denote the right-hand side of
(\ref{exp}) by $F(\infty, \infty)$, where the first parameter of
$F$ indicates the cardinality restriction over the subsets $N_i$'s
under consideration, whereas the second parameter of $F$ indicates
the maximum number of $N_i$'s allowed in the expansion. It is
straightforward that for fixed $Q$ and $K$,
\begin{equation}
F(\infty, \infty)=F(\infty, \infty)-F(Q, \infty)+F(Q, \infty)-F(Q,
K)+F(Q, K).
\end{equation}

We first examine $F(\infty, \infty)-F(Q, \infty)$. For every
subset $N$ of $\L'$, define
\begin{eqnarray} u_N=\left\{\begin{array}{ll}
w_N & \mbox{if $|N|\leq Q$};\\
0 & \mbox{otherwise}.\end{array} \right.
\end{eqnarray}
The difference in $F$ can then be interpreted as induced by
evaluating (\ref{exp}) using two sets of parameters $w_N$ and
$u_N$. By (\ref{avoid}), these two parameter sets both lie in the
region of analyticity of (\ref{exp}), thus intuitively, the
difference can be as small as desired when $Q$ is large enough. In
fact, it is bounded by $|R\cup W'|M^{|R\cup W'|}\epsilon(Q)$ by
the mean value theorem, applied to (\ref{avoid}) and
(\ref{nfinal}). Fix such a $Q$. We next examine $F(Q, \infty)-F(Q,
K)$. This difference can be regarded as the tail of the convergent
series (\ref{exp}), thus should also be small when $K$ is large
enough. We again refer to (\ref{avoid}) and (\ref{nfinal}), and
conclude that it is bounded by $|R\cup W'|M^{|R\cup
W'|}\epsilon(K)$. Fix such a $K$. For these two situations, the
only thing left to show now is that
\begin{equation}
\sum_{|R|\leq |W|P}|\tilde{w}_R|
\end{equation}
is finite, but this naturally follows from (\ref{suc}).

Finally, we examine $F(Q, K)$. By (\ref{w}) and (\ref{alpha}),
$w_N$ only depends on image sites in $N$. As $R\cup W'$ and
$\cup_1^p N_i$ is block-connected, $F(Q, K)$ will only depend on
image sites in a finite region (roughly a ball with radius
$S(|W|P+QK)$). If $Z$ is outside this region, then
\begin{equation}
\sum_{|R|\leq |W|P}\tilde{w}_R F(Q, K)
\end{equation}
is a constant with respect to $\sigma'_Z$, thus, when summing over
all possible image configurations $\sigma'$ as in (\ref{small}),
it vanishes.
\end{proof}

\begin{proposition}
\label{pen} Suppose the original interaction $J$ is at high
temperature (cf. (\ref{Jr})). Suppose also that it is finite-range
and translation-invariant. Then for subset $W$ of the original
lattice and subset $Z$ of the image lattice, as the distance $l(W,
Z)$ between $W$ and $Z$ gets large, the partial derivative
$|\frac{\partial J'(Z)}{\partial J(W)}|$ decays sub-exponentially,
a little slower than $\exp(-l(W, Z)^{1/2})$.
\end{proposition}

\begin{proof}
For notational convenience, we denote $l(W, Z)$ simply by $l$.
Take
\begin{equation}
P=\frac{1}{|W|}\left(\frac{l}{2S}\right)^{\alpha},
\end{equation}
and
\begin{equation}
Q=K=\left(\frac{l}{2S}\right)^{\beta},
\end{equation}
where $0<\alpha<\beta\leq 1/2$. We examine (\ref{enough}) closely.
The first factor,
\begin{equation*}
M^{D}\left(1+\log(M)\right)^{D},
\end{equation*}
is just a constant. The second factor is more complicated and thus
merits more attention. The first term, $\epsilon(P)/\log(M)$,
decays like $\exp\left(-l^{\alpha}\right)$, whereas the second
term,
\begin{equation*}
\left(\epsilon(Q)+\epsilon(K)\right)(1+P)DM^{(1+P)D},
\end{equation*}
decays like
$\exp\left(-l^{\beta}+l^{\alpha}\right)\sim\exp\left(-l^{\beta}\right)$.
Piecing it all together, $|\frac{\partial J'(Z)}{\partial J(W)}|$
decays sub-exponentially, like $\exp(-l^{\alpha})$.
\end{proof}

\section{Upper bound for the RG linearization}
\label{linearization}
\begin{definition}
For every subset $Z$ of the image lattice, the linearization
$\mathrm{L}(J)$ of the RG transformation for $J$ at high
temperature is given by a linear function of $K$ (which indicates
variation from infinite temperature),
\begin{equation}
\label{linear} \mathrm{L}(J)K(Z)=\sum_W \frac{\partial
J'(Z)}{\partial J(W)}K(W),
\end{equation}
where $W$ ranges over all finite subsets of the original lattice.
\end{definition}

\begin{proposition}
\label{ul} Consider finite-range and translation-invariant
Hamiltonians. Fix a subset $Z$ of the image lattice. Let $n(E)$ be
the number of subsets $W$ of the original lattice that are at most
$E$-distance away from $Z$,
\begin{equation}
n(E)=\#\{W: l(W, Z)\leq E\}.
\end{equation}
Then $n(E)$ grows polynomially in $E$.
\end{proposition}

\begin{proof}
Due to our finite-range and translation-invariant assumptions on
the Hamiltonian,
\begin{equation}
n=\sup_{y\in \L'} \#\{W: y\in W'\}<\infty.
\end{equation}
Thus $n(E)$ grows at the same rate as the volume of a
$d$-dimensional ball with radius $E$, i.e., polynomial growth
$E^d$.
\end{proof}

\begin{theorem}
\label{lin} Suppose the original interaction $J$ is at high
temperature (cf. (\ref{Jr})). Suppose also that it is finite-range
and translation-invariant. Then the linearization $\mathrm{L}(J)$
of the RG transformation (\ref{linear}) is well-defined and has an
upper bound.
\end{theorem}

\begin{proof}
This is mainly due to the fact that sub-exponential decay
dominates polynomial growth. Take $K$ with $||K||_r$ small. As
$||K||_{\infty}\leq ||K||_r$, $||K||_{\infty}$ is small also. By
Propositions \ref{pen} and \ref{ul},
\begin{eqnarray}
|\mathrm{L}(J)K(Z)|&\leq& \sum_{n=0}^{\infty}\sum_{n\leq l(W, Z)<
n+1} \left|\frac{\partial J'(Z)}{\partial
J(W)}\right|\left|K(W)\right|\\ &\lesssim&
||K||_{\infty}\sum_{n=0}^{\infty}\exp(-n^{\alpha})(n+1)^{d}.
\end{eqnarray}
Our claim then follows from the integral test.
\end{proof}

\section*{Acknowledgments}
This work was in partial fulfillment of the requirements for the
PhD degree at the University of Arizona. The author owes deep
gratitude to her PhD advisor Bill Faris for his continued help and
support. She also thanks Tom Kennedy, Doug Pickrell, and Bob Sims
for their kind and helpful suggestions and comments.

\end{document}